\documentclass[12pt,draftcls,peerreview,onecolumn]{IEEEtran}

\usepackage{cite}
\usepackage{subfigure}
\usepackage{graphicx}
\usepackage{epsfig,color}
\usepackage{url}
\usepackage{amsmath}
\usepackage{amssymb}
\newtheorem{theorem}{Theorem}
\newtheorem{lemma}[theorem]{Lemma}
\newtheorem{cor}[theorem]{Corollary}

\begin{document}

\title{On the Shannon capacity and queueing stability of random access multicast}

\author{Brooke~Shrader,~\IEEEmembership{Student Member, IEEE}
        and~Anthony~Ephremides,~\IEEEmembership{Fellow,~IEEE}
\thanks{B. Shrader and A. Ephremides are with the University of Maryland, College
Park. Portions of the material in this paper were presented at the
2005 and 2006 International Symposium on Information Theory~(ISIT)
and at the Fall 2006 Information Theory Workshop~(ITW).} }


\maketitle

\begin{abstract}
We study and compare the Shannon capacity region and the stable
throughput region for a random access system in which source nodes
multicast their messages to multiple destination nodes. Under an
erasure channel model which accounts for interference and allows for
multipacket reception, we first characterize the Shannon capacity
region. We then consider a queueing-theoretic formulation and
characterize the stable throughput region for two different
transmission policies: a retransmission policy and random linear
coding. Our results indicate that for large blocklengths, the random
linear coding policy provides a higher stable throughput than the
retransmission policy.
Furthermore, our results provide an example of a transmission policy
for which the Shannon capacity region strictly outer bounds the
stable throughput region, which contradicts an unproven conjecture
that the Shannon capacity and stable throughput coincide for random
access systems.
\end{abstract}


%

\IEEEpeerreviewmaketitle

\section{Introduction}




A fundamental question of communication theory is: at what rate can
information be transmitted reliably over a noisy channel? There is
more than one way to go about answering this question. For instance,
consider an erasure channel where the parameter $\epsilon$ denotes
the probability with which a transmission on the channel is lost;
with probability $1-\epsilon$ the transmission is received without
error. The traditional approach for describing the rate of reliable
communication for the erasure channel is to cite its Shannon
capacity, which is $1-\epsilon$ bits per channel use for a channel
with binary inputs and outputs. If feedback is available to notify
the sender when a channel input is erased, then the capacity can be
achieved by retransmitting lost inputs \cite{CoverThomas}.

Alternatively, the rate of reliable communication can be described
by the {\it maximum stable throughput}. In this setting, we view the
channel input as a packet of fixed length that arrives at a random
time to a sender and is stored in an infinite-capacity buffer while
awaiting its turn to be sent over the channel. The packets in the
buffer form a queue that is emptied in a specified order,
traditionally in first-in-first-out (FIFO) order. A transmission
protocol, for instance an Automatic Repeat Request (ARQ) protocol,
provides a form of redundancy to ensure that the packet is received
correctly. The maximum stable throughput is the highest rate (in
packets per channel use) at which packets arrive to the sender while
ensuring that the queue remains finite. In the case of the erasure
channel, if we assume that packets arrive to the sender according to
a Bernoulli process, feedback is available to notify the sender of
lost packets, and the transmission protocol consists of
retransmitting lost packets, then the buffer at the sender forms a
discrete-time $M/M/1$ queue with maximum departure rate
$1-\epsilon$. The maximum stable throughput is $1-\epsilon$, which
is identical to the Shannon capacity.

For {\it multiuser systems}, the relation between the Shannon
capacity and the maximum stable throughput for communication over an
erasure channel has been explored in the context of random multiple
access, where user $n$ transmits with probability $p_n$ at each
transmission opportunity. This form of random channel access is a
variation on Abramson's ALOHA protocol \cite{Abramson70} and is
particularly attractive for use in mobile ad-hoc networks because it
is robust to variations in the network topology and can be
implemented without coordination among the transmitting nodes.

The finite-user, buffered random access problem was first formulated
by Tsybakov and Mikhailov \cite{TsybakovMikhailov}, who provided a
sufficient condition for stability and thus a lower bound on the
maximum stable throughput. The problem they considered was a system
in which finitely many source nodes with infinite-capacity queues
randomly access a shared channel to send messages to a central
station. Feedback was used to notify the source nodes of failed
transmissions and a retransmission scheme was used to ensure
eventual successful reception at the central station. The users were
assumed to access a {\it collision channel}, in which transmission
by more than one source results in the loss of all transmitted
packets with probability 1. This collision channel model is
equivalent to an erasure channel, where for user $n$, $1-\epsilon =
p_n \prod_{j \neq n} (1-p_j)$. Further progress on this problem was
made in \cite{RaoAE88}, in which stochastic dominance arguments were
explicitly introduced to find the stable throughput region for a
system with 2 source nodes, and in \cite{Szpa94}, wherein a stable
throughput region based on the joint queue statistics was found for
3 source nodes. An exact stability result for arbitrarily (but
finitely) many users has not been found, but bounds have been
obtained in \cite{Szpa94} and \cite{LuoAE99}. Recently, the authors
of \cite{NawareEtAl03} improved upon the collision channel model
used in all previous works and studied a channel with multipacket
reception (MPR) capability. They showed that the stable throughput
region transfers from a non-convex region under the collision
channel model to a convex region bounded by straight lines for a
channel with MPR.

The Shannon capacity region of a random access system was considered
in \cite{MasseyMathys} and \cite{Hui}, which both obtained the
capacity region for finitely many source nodes transmitting to a
central station under the collision channel model. That capacity
result can be viewed as the capacity of an asynchronous multiple
access channel, which was obtained in \cite{Poltyrev83} and
\cite{HuiHumblet85}. The more recent contribution of \cite{Tinguely}
shows explicitly how the random access capacity region in
\cite{MasseyMathys, Hui} is obtained from the results in
\cite{Poltyrev83, HuiHumblet85}, in addition to analyzing the
capacity for a channel in which packets involved in a collision can
be recovered.

It was noted by Massey and Mathys \cite{MasseyMathys} and Rao and
Ephremides \cite{RaoAE88} that the stable throughput and Shannon
capacity regions coincide for the special cases of two source nodes
and infinitely-many source nodes. As with the point-to-point erasure
channel, this result is surprising in that it suggests that the
bursty nature of arriving packets, which is captured in the
stability problem but not in the capacity problem, is immaterial in
determining the limits on the rate of reliable communication. It has
been conjectured that the stable throughput region for finitely-many
source nodes transmitting to a central station (which is an unsolved
problem) coincides with the corresponding capacity region. This
conjecture was explored in \cite{Anantharam}, in which it was shown
to hold in a special case involving correlated arrivals of packets
at the source nodes. Recently, further progress was made in
\cite{RockeyAE_ITTrans06} towards showing that the stable throughput
and Shannon capacity regions coincide for transmission over a
channel with MPR. However, a complete proof has still not been
found.

In this work we explore the relation between the stable throughput
and Shannon capacity for a random access system in which source
nodes {\it multicast} their messages to multiple receivers.
Specifically, we consider a system with two source nodes and two
destination nodes. The source nodes randomly access the channel to
transmit to the two destination nodes. We first characterize the
Shannon capacity region of the system, which is similar to
Ahlswede's result \cite{Ahlswede74} on the capacity for two senders
and two receivers. We then move to the queueing stability problem
and characterize the stable throughput region for our random access
multicast system with two different transmission policies. The first
is a retransmission policy, which is the policy used in previous
works on queueing stability for a (single) central station. Next, we
study random access multicast with random linear coding, which is
inspired by the recent development of network coding
\cite{NetCod1,NetCod2} and fountain coding
\cite{Luby02,Shokrollahi06}. Our results show that for multicast
transmission, the maximum stable throughput under the retransmission
policy does not reach the Shannon capacity, however, the random
linear coding scheme provides a maximum stable throughput that
asymptotically approaches the Shannon capacity.

\section{System model} \label{section:Model}

\begin{figure}
\centering
\includegraphics[width=0.3\textwidth]{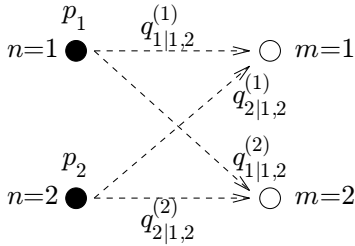}
\caption{The multicast scenario we consider in this work. Two source
nodes randomly access the channel to multicast to two destination
nodes. When both sources transmit, which happens in a slot with
probability $p_1p_2$, the reception probabilities are as shown
above.} \label{fig:general_2sour2dest}
\end{figure}

The system we consider throughout this work is shown in Fig.
\ref{fig:general_2sour2dest}. Two source nodes, indexed by $n$, each
generate messages to be multicast to two destination nodes, indexed
by $m$, over a shared channel. The data generated at source $n=1$ is
assumed independent of the data generated at source $n=2$. Time is
slotted; one time slot corresponds to the amount of time needed to
transmit a single packet over the shared channel. A packet is a
fixed-length vector of bits. In each time slot, if source $n$ has a
packet to transmit, then we refer to the source as being {\it
backlogged}; otherwise the source is {\it empty}. A backlogged
source transmits in a slot with probability $p_n$. We refer to $p_n$
as the transmission probability; it encapsulates random access to
the channel. We assume the value of $p_n$ to be fixed in time.
(I.e., we do not assume retransmission control, in which $p_n$ is
varied over time according to the history of successful
transmissions.)

The channel model we consider is similar to the model used in
\cite{NawareEtAl03}. A transmitted packet is received without error
with a certain probability. Otherwise, the packet is lost and cannot
be recovered. We assume that the channels between different
source-destination pairs are independent. We introduce the following
{\it reception probabilities} for sources $n=1,2$ and destinations
$m=1,2$.
\begin{equation}
q_{n|n}^{(m)}  = \mbox{Pr}\{\mbox{packet from $n$ is received at
$m$} | \mbox{ only $n$ transmits}\}
\end{equation}
\begin{equation}
q_{n|1,2}^{(m)}  = \mbox{Pr}\{\mbox{packet from $n$ is received at
$m$} | \mbox{ both sources transmit}\}
\end{equation}
We assume throughout that interference cannot increase the reception
probability on the channel, i.e., $q_{n|n}^{(m)}
> q_{n|1,2}^{(m)}$. The reception probabilities inherently account
for interference and also allow for multipacket reception (MPR).
Note that these probabilities can capture the effects of fading on
the wireless channel by setting them equal to the probability that a
fading signal, represented by a random variable, exceeds a certain
signal to interference plus noise (SINR) threshold. The collision
channel model used in a number of previous works is given by
$q_{n|n}^{(m)} = 1$, $q_{n|1,2}^{(m)}=0$.

\section{Shannon capacity region}

We first analyze the system under consideration in an
information-theoretic framework to determine the Shannon capacity
region. Thus we assume in this section that the sources are always
backlogged. Coding is performed on the data at the source nodes, and
we view each transmitted packet as a codeword symbol, where a
collection of $N$ codeword symbols constitutes a codeword. In order
to allow for random access of the channel, we assume that while
codeword symbols are synchronized (equivalently, time is slotted),
the codewords for the two sources do not necessarily begin and end
in the same time slots. We analyze this system first for a general
discrete memoryless channel (DMC) and then apply our random access
channel model to obtain the capacity region.

\subsection{Discrete memoryless channel}

The discrete memoryless channel we consider consists of discrete
alphabets $\mathcal X_1, \mathcal X_2, \mathcal Y_1,$ and  $\mathcal
Y_2$ and transition probability matrix $W(y_1,y_2|x_1,x_2)$. The
channel can be decomposed into two multiple access channels, each
corresponding to a destination node and defined as follows.
\begin{eqnarray}
W_1(y_1 | x_1, x_2) & = & \sum_{y_2 \in {\cal Y}_2} W(y_1, y_2| x_1, x_2) \\
W_2(y_2 | x_1, x_2) & = & \sum_{y_1 \in {\cal Y}_1} W(y_1, y_2| x_1,
x_2)
\end{eqnarray}
We assume that there is no feedback available on the channel, which
differs from the assumption made in the model for the stable
throughput problem presented in Section
\ref{section:StableThroughputRegion}.

Source node $n$ generates a sequence of messages
$J_n^1,J_n^2,\hdots$ where the $t^{th}$ message $J_n^t$ takes values
from the message set $\{1,2,\hdots, 2^{NR_n}\}$. The messages are
chosen uniformly from $\{1,2,\hdots, 2^{NR_n}\}$ and independently
over $n$. The encoding function $f_n$ at source $n$ is given by the
mapping
\begin{equation}
f_n: \{1,2,\hdots, 2^{NR_n}\} \to \mathcal X_n^N, \quad n=1,2.
\end{equation}
The encoder output consists of a sequence of codewords $X_n(J_n^t)$,
$t \geq 1$. The system is asynchronous in the following sense. Each
source and each destination maintain a clock. Let $S_{nm}$ denote
the amount of time that the clock at source $n$ is running ahead of
the clock at destination $m$. The $S_{nm}$ are assumed to be
integers, meaning that time is discrete and transmissions are
symbol-synchronous. The time at each clock can be divided into {\it
periods} of length $N$ corresponding to the length of a codeword.
Let $D_{nm}$ denote the offset between the start of periods at
source $n$ and destination $m$, where $0 \leq D_{nm} \leq N-1$. We
assume that $D_{nm}$ are uniform over $[0,1,\hdots,N-1]$ for all
$N$. The codeword $X_n(J_n^1)$ is sent at time $0$ on the clock at
source $n$.

A sequence of channel outputs are observed at each destination,
where the outputs at destination $m$ each take values from the
alphabet $\mathcal Y_m$. The decoder operates on a sequence of
$N(T+1)$ channel outputs to form an estimate of a sequence of $T+1$
messages. A decoder is defined as follows.
\begin{equation}
\phi_{nm}: \mathcal Y_m^{N(T+1)} \to \{1,2,\hdots,2^{NR_n} \}^{T+1},
\quad n,m=1,2
\end{equation}
where $\{1,2,\hdots,2^{NR_n} \}^{T+1}$ denotes the $(T+1)$-fold
Cartesian product of $\{1,2,\hdots,2^{NR_n}\}$. Since the decoder
must synchronize on a particular source $n$, the decoding function
is defined separately for each source. The output of the decoder is
a sequence of message estimates $\hat{J}_{nm}^1,
\hat{J}_{nm}^2,\hdots, \hat{J}_{nm}^{T+1}$, where $\hat{J}_{nm}^t$
denotes the estimate at destination $m$ of the $t^{th}$ message sent
by source $n$. The error criterion we consider is the average
probability of error $P_e^t$ defined as
\begin{equation}
P_e^t = \mbox{Pr} \left\{ \bigcup_{m} \bigcup_{n} \{ J_n^t \neq
\hat{J}_{nm}^t \} \right\}.
\end{equation}
The rate pair $(R_1, R_2)$ is {\it achievable} if there exists
encoding functions $(f_1,f_2)$ and decoding functions
$(\phi_{11},\phi_{12}, \phi_{21}, \phi_{22})$ such that $P_e^t
\rightarrow 0$ for all $t$ as $N \rightarrow \infty$. The capacity
region is the set of all achievable rate pairs

The model we consider here is a compound version of the {\it totally
asynchronous multiple access channel} treated in \cite{Poltyrev83}
and \cite{HuiHumblet85}. As shown in those works, the asynchrony in
the system results in the lack of a convex hull operation, and this
holds as well in our compound version of the problem. The capacity
region is presented below and the proof of the theorem is described
in Appendix \ref{app:proofcapreg}.

\begin{theorem}The capacity region of the asynchronous compound multiple
access channel is the closure of all rate points $(R_1,R_2)$ that
lie in the region
\begin{equation*}
\setlength{\nulldelimiterspace}{0pt}
\begin{IEEEeqnarraybox}[\relax][c]{c} \bigcap_{m}
\left\{(R_1,R_2):\begin{IEEEeqnarraybox}[\relax][c]{r'l'l}
R_1 & \!\! < \!\!& I(X_1;Y_m | X_2) \\
R_2 & \!\! < \!\!& I(X_2;Y_m | X_1) \\
R_1+R_2 & \!\! < \!\!& I(X_1, X_2;Y_m) \end{IEEEeqnarraybox}
\right\} \end{IEEEeqnarraybox}
\end{equation*}
for some product distribution $P(x_1)P(x_2)W$. \label{thrm:capreg}
\end{theorem}

\subsection{Random access system}

We now turn our attention to a random access system and apply
Theorem \ref{thrm:capreg} to determine the capacity region of the
system. Each codeword symbol corresponds to a packet transmitted
over the channel shown in Figure \ref{fig:general_2sour2dest}. We
define the common input alphabet as ${\cal X}=\{0, 1,
2,\hdots,2^u\}$, where $X_n \in {\cal X}$, for $n=1,2$. A channel
input $X_n$ can be either a packet of length $u$ bits (an
information-bearing symbol) or an idle symbol. The $0$ symbol is the
idle symbol and we let Pr$\{ X_n = 0\}=1-p_n$ according to the
random access transmission probability. We assume a uniform
distribution on the information-bearing codeword symbols,
$\mbox{Pr}\{ X_n = x\}=p_n/2^u$, $x=1,2,\hdots 2^u$, meaning that a
packet is equally likely to be any sequence of $u$ bits. The channel
output at receiver $m$ is given by $Y_m=(Y_{1m},Y_{2m}) \in {\cal
X}' \times {\cal X}'$ where $Y_{nm}$ denotes the packet from source
$n$ and ${\cal X}'={\cal X} \cup \Delta$. The $\Delta$ symbol
denotes a packet in error.

The introduction of the idle symbol $0$ results in additional
protocol or timing information being transmitted over the channel.
The information content of this idle symbol is $h_b(p_n), n=1,2$,
where $h_b$ denotes the binary entropy function. The term $h_b(p_n)$
appears in the proof provided below and represents the protocol
information that is studied by Gallager in \cite{Gallager76}.
Because we would like our capacity result to represent the rate of
reliable communication of data packets, we will aim to exclude this
timing information. We do so by considering capacity in packets/slot
in the limit as $u \rightarrow \infty$, meaning that the data
packets grow large and the fraction of timing information
transmitted approaches $0$. The timing information is excluded in
previous work on random access capacity. In \cite{MasseyMathys},
prior to the start of transmission, a ``protocol sequence''
indicating the occurrence of idle slots is generated at the source
and communicated to the receiver, effectively eliminating timing
information. In \cite{Tinguely}, the capacity for $u \rightarrow
\infty$ is presented. The capacity of the random access multicast
system is given in the following Corollary to Theorem
\ref{thrm:capreg}.

\begin{cor}
The capacity region of the random access system with two sources and
two destinations is the closure of $(R_1,R_2)$ for which
\begin{eqnarray*}
R_1 & < &  \min_{m=1,2} p_1(1-p_2)q_{1|1}^{(m)}+p_1p_2q_{1|1,2}^{(m)} \\
R_2 & < & \min_{m=1,2}(1-p_1)p_2q_{2|2}^{(m)}+p_1p_2q_{2|1,2}^{(m)}
\end{eqnarray*}
for some $(p_1,p_2) \in [0,1]^2$.
\label{cor:capacity}
\end{cor}

\begin{proof}
The result follows by applying the assumptions about the input
distribution and channel reception probabilities to the expressions
given in Theorem \ref{thrm:capreg}. We first solve for
$I(X_1;Y_m|X_2)$ by conditioning on $X_2$ to obtain
\begin{equation}
I(X_1;Y_m|X_2) = (1-p_2)I(X_1;Y_m|X_2=0) + p_2I(X_1;Y_m|X_2\neq0).
\end{equation}
An expression for $I(X_1;Y_m|X_2=0)$ can be found from the following
sequence of equalities.
\begin{eqnarray}
I(X_1;Y_m|X_2=0) & = & H(X_1|X_2 = 0) - H(X_1|Y_m,X_2 = 0) \nonumber \\
 & \stackrel{(a)}= & H(X_1) -  \mbox{Pr}(Y_{1m}=\Delta|X_2 = 0) \log_2 2^u \nonumber \\
 & = & -(1-p_1)\log_2 (1-p_1) - p_1 \log_2 (p_1/2^u) - p_1
 (1-q_{1|1}^{(m)}) u \nonumber \\
 & = & h_b(p_1) + p_1 \log_2 2^u - p_1 u + p_1 q_{1|1}^{(m)}u \nonumber \\
 & = & h_b(p_1) + p_1 q_{1|1}^{(m)} u \label{eqn:HX1GivenX2_1}
\end{eqnarray}
where $(a)$ holds since $X_2$ is independent of $X_1$ and since
$H(X_1|Y_{1m}\neq \Delta,X_2 = 0)=0$. For $(X_1;Y_m|X_2 \neq 0)$ we
have
\begin{eqnarray}
I(X_1;Y_m|X_2 \neq 0) & = & H(X_1|X_2 \neq 0) - H(X_1|Y_m,X_2 \neq 0) \nonumber \\
 & = & H(X_1) - \mbox{Pr}(Y_{1m}=\Delta|X_2 \neq 0) \log_2 2^u  \nonumber \\
 & = & -(1-p_1)\log_2 (1-p_1) - p_1 \log_2 (p_1/2^u) - p_1
 (1-q_{1|1,2}^{(m)}) u \nonumber \\
 & = & h_b(p_1) + p_1 \log_2 2^u - p_1 u  + p_1 q_{1|1,2}^{(m)}u \nonumber \\
 & = & h_b(p_1) + p_1 q_{1|1,2}^{(m)} u. \label{eqn:HX1GivenX2_2}
\end{eqnarray}
Combining expressions (\ref{eqn:HX1GivenX2_1}) and
(\ref{eqn:HX1GivenX2_2}) results in
\begin{equation}
I(X_1;Y_m|X_2) = h_b(p_1) + u p_1 (1-p_2)q_{1|1}^{(m)} + u p_1p_2
q_{1|1,2}^{(m)} \mbox{ bits/transmission}.
\end{equation}
Since one packet corresponds to $u$ bits, we divide by $u$ to obtain
a result in units of packets per slot. We then let $u \rightarrow
\infty$ to obtain
\begin{equation}
I(X_1;Y_m|X_2) =   p_1 (1-p_2)q_{1|1}^{(m)} + p_1p_2 q_{1|1,2}^{(m)}
\mbox{ packets/slot}.
\end{equation}
By following a similar approach, we can show that
$I(X_2;Y_m|X_1)$ is given as follows.
\begin{equation}
I(X_2;Y_m|X_1) = (1-p_1)p_2q_{2|2}^{(m)} + p_1p_2 q_{2|1,2}^{(m)}
\mbox{ packets/slot}
\end{equation}
The bound on the sum rate can be found by breaking up
$I(X_1,X_2;Y_m)$ into four terms.
\begin{eqnarray}
I(X_1,X_2;Y_m) & = & H(X_1,X_2) - H(X_1,X_2|Y_m) \nonumber \\
 & = & H(X_1|X_2) + H(X_2) - H(X_1|Y_m,X_2) - H(X_2|Y_m) \nonumber \\
 & = & H(X_1) + H(X_2) - H(X_1|Y_{1m},X_2) - H(X_2|Y_{1m},Y_{2m})
 \label{eqn:IX1X2;Y4terms}
\end{eqnarray}
For $n=1,2$, the terms $H(X_n)$ can be expressed as
\begin{eqnarray}
H(X_n) & = & -(1-p_n)\log_2(1-p_n) - p_n\log_2(p_n/2^u) \nonumber \\
 & = & h_b(p_n) + p_n\log_2 2^u \nonumber \\
 & = & h_b(p_n) + p_n u. \label{eqn:HXn}
\end{eqnarray}
The last two terms in (\ref{eqn:IX1X2;Y4terms}) can be found in the
following manner.
\begin{eqnarray}
H(X_1|Y_{1m},X_2) & = & H(X_1|Y_{1m}=\Delta,X_2) \mbox{Pr}(Y_{1m}=\Delta|X_2) \nonumber \\
 & = & (1-p_2) H(X_1|Y_{1m}=\Delta,X_2=0)\mbox{Pr}(Y_{1m}=\Delta|X_2 =0) \nonumber \\
 & & + p_2 H(X_1|Y_{1m}=\Delta,X_2 \neq
 0) \mbox{Pr}(Y_{1m}=\Delta|X_2 \neq 0) \nonumber \\
 & = & (1-p_2) \log_2 2^u \mbox{Pr}(Y_{1m}=\Delta|X_2=0) + p_2
 \log_2
 2^u \mbox{Pr}(Y_{1m}=\Delta|X_2 \neq 0) \nonumber \\
 & = & u (1-p_2) p_1 (1-q_{1|1}^{(m)}) + u p_2 p_1 (1-q_{1|1,2}^{(m)})
 \nonumber \\
 & = & u p_1 - u p_1 (1-p_2)q_{1|1}^{(m)} - u p_1 p_2 q_{1|1,2}^{(m)}
 \label{eqn:HX1GivenY1X2}
\end{eqnarray}
\begin{eqnarray}
H(X_2|Y_{1m},Y_{2m}) & = & (1-p_1) H(X_2|Y_{1m}=0,Y_{2m}) + p_1
H(X_2|Y_{1m} \neq 0,Y_{2m}) \nonumber \\
 & = & (1-p_1) H(X_2|Y_{1m}=0,Y_{2m}=\Delta)\mbox{Pr}(Y_{2m}=\Delta|Y_{1m}=0) \nonumber \\ & & + p_1
H(X_2|Y_{1m} \neq 0,Y_{2m}=\Delta) \mbox{Pr}(Y_{2m}=\Delta|Y_{1m} \neq 0) \nonumber \\
 & = & (1-p_1) \log_2 2^u p_2 (1-q_{2|2}^{(m)}) + p_1 \log_2 2^u p_2
 (1-q_{2|1,2}^{(m)}) \nonumber \\
 & = & u p_2 - u(1-p_1)p_2q_{2|2}^{(m)} - u p_1 p_2 q_{2|1,2}^{(m)}
 \label{eqn:HX2GivenY}
\end{eqnarray}
By substituting (\ref{eqn:HXn}), (\ref{eqn:HX1GivenY1X2}), and
(\ref{eqn:HX2GivenY}) in (\ref{eqn:IX1X2;Y4terms}), dividing by $u$
and taking $u \rightarrow \infty$, we obtain
\begin{equation}
I(X_1,X_2;Y_m) = p_1 (1-p_2)q_{1|1}^{(m)} +  p_1p_2 q_{1|1,2}^{(m)}
+ (1-p_1)p_2q_{2|2}^{(m)} +  p_1p_2 q_{2|1,2}^{(m)} \mbox{
packets/slot} .
\end{equation}
Since $I(X_1,X_2;Y_m) = I(X_1;Y_m|X_2) + I(X_2;Y_m|X_1)$, the bound
on the sum rate is superfluous in terms of describing the capacity
region. The result follows.
\end{proof}

\section{Stable throughput region}
\label{section:StableThroughputRegion}

In this section we treat the system shown in Fig.
\ref{fig:general_2sour2dest} as a network of queues and state the
stable throughput region of the system, which is a generalization of
previous results on the stable throughput for a system with a single
destination node. The model we consider is as follows. We no longer
assume that the sources are always backlogged; instead a random
Bernoulli process with rate $\lambda_n$ packets/slot, $n=1,2$ models
the arrival of packets at each source. Packets that are not
immediately transmitted are stored in an infinite-capacity buffer
maintained at each source. Transmissions occur according to the
random access protocol with source $n$ transmitting with probability
$p_n$ when it is backlogged. If a source is empty in a given slot,
it does not access the channel. Each source-destination pair is
assumed to maintain an orthogonal feedback channel so that
instantaneous and error-free acknowledgements can be sent from the
destinations to the sources. We find the stable throughput region
for two different transmission schemes: retransmissions and random
linear coding.

The queue at each source is described by its arrival process and
departure process. As $\lambda_n$ represents the arrival rate, we
let $\mu_n$ denote the departure or service rate. In general, a
queue is said to be stable if departures occur more frequently than
arrivals, i.e., $\lambda_n < \mu_n$. This statement is made precise
by Loynes' result \cite{Loynes62}, which states that if the arrival
and departure processes are non-negative, finite, and strictly
stationary, then $\lambda_n < \mu_n$ is a necessary and sufficient
condition for stability. Stability of the queue is equivalent to
ergodicity of the Markov chain representing the queue length. The
stable throughput region for a given transmission policy is defined
as the set of all $(\lambda_1,\lambda_2)$ for which there exists
transmission probabilities $(p_1,p_2)$ such that both queues remain
stable.

The difficulty in finding the stable throughput region for our
system (and for any buffered random access system) arises from the
interaction of the queues. In particular, the service rate $\mu_n$
of source $n$ will vary according to whether the {\it other} source
is empty or backlogged and can create interference on the channel.
To overcome this difficulty, the technique provided in
\cite{RaoAE88} of introducing a dominant system can be used to
decouple the sources. In a dominant system, one of the sources
behaves as if it is always backlogged by transmitting ``dummy''
packets when it empties. The queue length in a dominant system
stochastically dominates (i.e., is never smaller than) the queue
length in the system of interest, meaning that stability in a
dominant system implies stability in the original system. Since one
source always behaves as if it is backlogged, the service rate in
the dominant system can be easily found. Let $\mu_{nb}$ denote the
service rate at source $n$ when the other source is backlogged and
$\mu_{ne}$ the service rate when the other source is empty. Using
the dominant systems approach, the stable throughput region for a
system with two sources can be found exactly. This region is stated
in the following theorem, which is a generalization of the result in
\cite{RaoAE88}. Note that in the stable throughput region presented
below, the service rates $\mu_{nb}$ and $\mu_{ne}$ are functions of
$p_n$, $n=1,2$ (although not explicitly shown in these expressions).

\begin{theorem} \cite{RaoAE88} \label{thrm:2x2stabreg} For a network with two sources and two destinations,
the stable throughput region is given by the closure of ${\cal
L}(p_1,p_2)$ where
\begin{equation}
{\cal L}(p_1,p_2) = \bigcup_{i=1,2} {\cal L}_i (p_1,p_2)
\end{equation}
and
\begin{equation*}
\setlength{\nulldelimiterspace}{0pt} {\cal L}_1(p_1,p_2) =
\left\{(\lambda_1,\lambda_2):\begin{IEEEeqnarraybox}[\relax][c]{r'l'l}
\lambda_1 & \!\! < \!\!& \frac{\lambda_2}{\mu_{2b}} \mu_{1b} +
\left( 1 -\frac{\lambda_2}{\mu_{2b}}\right) \mu_{1e} \\
\lambda_2 & \!\! < \!\!& \mu_{2b} \end{IEEEeqnarraybox} \right\}
\end{equation*}
\begin{equation*}
\setlength{\nulldelimiterspace}{0pt} {\cal L}_2(p_1,p_2) =
\left\{(\lambda_1,\lambda_2):\begin{IEEEeqnarraybox}[\relax][c]{r'l'l}
\lambda_1 & \!\! < \!\!& \mu_{1b} \\
\lambda_2 & \!\! < \!\!& \frac{\lambda_1}{\mu_{1b}} \mu_{2b} +
\left( 1 - \frac{\lambda_1}{\mu_{1b}} \right) \mu_{2e}
\end{IEEEeqnarraybox} \right\}
\end{equation*}
for some $(p_1,p_2) \in [0,1]^2$.
\hspace{0.1cm} 
\end{theorem}

In addition to the stable throughput region, we will be interested
in the throughput region of the random access system shown in Fig.
\ref{fig:general_2sour2dest}. The throughput region is the closure
of the service rates $\mu_{nb}$ for all $p_n$, $n=1,2$, where both
sources are assumed to be backlogged. In finding the throughput
region, there is no interaction between the sources, and the problem
is simpler than finding the stable throughput problem. Previous work
on buffered random access systems has shown that in all cases in
which the stable throughput region has been found, it is known to
coincide with the throughput region, suggesting that the empty state
of the buffer is insignificant in determining stability. The
relation between stable throughput and throughput is explored in
\cite{RockeyAE_ITTrans06}. Additionally, in
\cite{ShraderEphremidesITTrans07StabThrpt}, it is proved that for a
random access system with two sources and two destinations, the
stable throughput region coincides with the throughput region. This
result is restated below.

\begin{theorem} \cite{ShraderEphremidesITTrans07StabThrpt}
The stable throughput region of the random access system with two
sources and two destinations is equivalent to the throughput region,
which is given by the closure of $(\lambda_1,\lambda_2)$ for which
\begin{equation*}
\lambda_1  <  \mu_{1b}, \quad \lambda_2  <  \mu_{2b}
\end{equation*}
for some $(p_1,p_2) \in [0,1]^2$. \label{theorem:throughput}
\end{theorem}

We derive the backlogged and empty service rates $\mu_{nb}$ and
$\mu_{ne}$ for two different transmission schemes. Together with
Theorems \ref{thrm:2x2stabreg} and \ref{theorem:throughput}, this
provides us a complete characterization of the stable throughput
region.

\section{Stable throughput: retransmissions}
\label{section:StabilityRetrans}

In this section we describe the stable throughput region under
assuming that a retransmission protocol is used to ensure reliable
communication. In the retransmission scheme, as long as source $n$
has not received feedback acknowledgements from both destinations
$m=1,2$, it will continue to transmit the packet over the channel
with probability $p_n$ in each slot. As soon as the source has
received acknowledgements from both destinations, it will remove the
packet from its queue and begin transmitting the next packet waiting
in its buffer, if any. Let random variable $T_n$ denote the service
time for source $n$, given by the total number of slots that
transpire before the packet from source $n$ is successfully received
at both destinations. (Note that the service time includes slots
during which the source does not transmit, which happens with
probability $1-p_n$). Since each completed service in the
retransmission scheme results in 1 packet being removed from the
queue, the average service rate is given by $\mu_n = 1/E[T_n]$.

We first find the backlogged service rates $\mu_{nb}$. Let
$T_n^{(m)}$ denote the number of slots needed for successful
reception of a packet from source $n$ at destination $m$, $m=1,2$.
The $T_n^{(m)}$ are geometrically distributed according to the
transmission probabilities $p_n$ and reception probabilities
$q_{n|n}^{(m)}$, $q_{n|1,2}^{(m)}$. We introduce the following
notation. Let $\phi_n$ denote the probability of successful
reception of the packet from source $n$ at destination 1 given that
source $n$ transmits and that both sources are backlogged.
Similarly, $\sigma_n$ denotes the probability of successful
reception at destination 2 given that both sources are backlogged
and that source $n$ transmits. For instance, $\phi_1$ and $\sigma_1$
are given by
\begin{eqnarray}
\phi_1 & = & \overline{p}_2 q_{1|1}^{(1)} + p_2q_{1|1,2}^{(1)} \\
\sigma_1 & = & \overline{p}_2 q_{1|1}^{(2)} + p_2q_{1|1,2}^{(2)}
\end{eqnarray}
where $\overline{p}_n=1-p_n$. When source 2 is backlogged,
$T_1^{(1)}$ is geometrically distributed with parameter $p_1\phi_1$
and $T_1^{(2)}$ is geometrically distributed with parameter
$p_1\sigma_1$. The total service time for source 1 when source 2 is
backlogged will be given by the maximum of the service times to each
destination,
\begin{equation}
T_1 \sim \max_m T_1^{(m)}. \label{eqn:RetransServTime1B}
\end{equation}
Similarly, when source 1 is backlogged, the service time for source
2 is given by
\begin{equation}
T_2 \sim \max_m T_2^{(m)}, \label{eqn:RetransServTime2B}
\end{equation}
where
\begin{equation}
T_2^{(1)} \sim \mbox{geom}\left(p_2 \phi_2 \right), \quad T_2^{(2)}
\sim \mbox{geom}\left( p_2 \sigma_2 \right).
\label{eqn:RetransServTime2Bm}
\end{equation}
The expected maximum value $E[T_n]$ can be readily found and the
backlogged service rates are given by
\begin{equation}
\mu_{nb}=\frac{p_n\phi_n\sigma_n(\phi_n+\sigma_n-\tau_n)}{(\phi_n+\sigma_n)(\phi_n+\sigma_n-\tau_n)-\phi_n\sigma_n}
\label{eqn:mu_nb}
\end{equation}
where $\tau_n$ denotes the probability that a packet sent from $n$
is received at both destinations given that source $n$ transmits,
e.g., $\tau_1 = \overline{p}_2 q_{1|1}^{(1)}q_{1|1}^{(2)} + p_2
q_{1|1,2}^{(1)}q_{1|1,2}^{(2)}$.

The empty service rates $\mu_{ne}$ can be found directly from
$\mu_{nb}$ as
\begin{equation}
\mu_{1e}=\mu_{1b}|_{p_2=0}, \quad \mu_{2e}=\mu_{2b}|_{p_1=0}.
\label{eqn:mu_ne}
\end{equation}
The backlogged and empty service rates for random access multicast
with retransmissions have also been found in
\cite{ShraderEphremidesITTrans07StabThrpt}. In that work, a
different approach is used in which a Markov chain is used to model
which destinations have received the packet currently under
transmission. The result is identical to the one presented above.

We now compare the stable throughput region for the retransmissions
scheme to the Shannon capacity region. We consider the backlogged
service rate $\mu_{nb}$, since, by Theorem \ref{theorem:throughput},
this is the term that bounds the stable throughput for source $n$.
The expected service time for source 1 when source 2 is backlogged
is bounded as
\begin{eqnarray}
E[T_1] = E[\max_m T_1^{(m)}] & \stackrel{(a)} \geq & \max_m
E[T_1^{(m)}] \\
 & \stackrel{(b)} = & \max_m \frac{1}{p_1 (1-p_2) q_{1|1}^{(m)} + p_1 p_2
q_{1|1,2}^{(m)}} \\
 & = & \frac{1}{\min_m p_1 (1-p_2) q_{1|1}^{(m)} + p_1 p_2
q_{1|1,2}^{(m)}}
\end{eqnarray}
where $(a)$ follows from Jensen's inequality and $(b)$ follows from
the expected value of a geometrically distributed random variable.
Then the backlogged service rate $\mu_{1b}$ is bounded as
\begin{equation}
\mu_{1b} = \frac{1}{E[T_1]} \leq \min_m p_1 (1-p_2) q_{1|1}^{(m)} +
p_1 p_2 q_{1|1,2}^{(m)} \label{eqn:BoundRetransServRate}
\end{equation}
and $\mu_{2b}$ can be bounded similarly. Note that the right-hand
side of (\ref{eqn:BoundRetransServRate}) is equal to the bound on
the Shannon achievable rate $R_1$ given in Corollary
\ref{cor:capacity}. Thus we should expect that the Shannon capacity
region outer bounds the stable throughput region for the
retransmission scheme.

\section{Stable throughput: random linear coding}

In this section we present two different approaches to analyzing the
stable throughput region for a transmission scheme in which groups
of $K$ packets at the front of the queue are randomly encoded and
transmitted over the channel. By encoded, we mean that a random
linear combination of the $K$ packets is formed, and we refer to
this random linear combination as a coded packet. As soon as a
destination has received enough coded packets and is able to decode
the original $K$ packets, it does so and sends an acknowledgement to
the source over its feedback channel. Once the source receives
acknowledgements from both destinations, it removes the $K$ packets
it has been encoding and transmitting from its queue and begins
encoding and transmission of the next $K$ packets waiting in its
buffer. The stable throughput region for a similar system with $K=2$
is found in \cite{SagduyuAE06}.

For a given source, let $s_1,s_2,\hdots,s_K$ denote the $K$ (binary)
packets at the front of the queue. Random linear coding is performed
in the following manner. A coded packet is randomly generated as
\begin{equation}
\sum_{i=1}^{K} a_i s_i \label{eqn:RLCCodedPacket}
\end{equation}
where the coefficients $a_i$ are generated according to
$\mbox{Pr}(a_i=1)=1-\mbox{Pr}(a_i=0)=1/2$, $i=1,2,\hdots,K$, and
$\sum$ denotes modulo-2 addition. A coded packet formed in this way
is the same length as one of the original packets $s_i$ and can be
transmitted over the channel within the same time period of a slot.
For each coded packet sent, the coefficients $a_i$ used in
generating the packet will be appended to the header of the packet.
In each slot, the source generates a coded packet according to
(\ref{eqn:RLCCodedPacket}) and sends it over the channel. Regardless
of whether a previously-generated coded packet has been received,
the source generates a new coded packet in each slot and transmits
it with probability $p_n$.

\subsection{Expected service time}

We first analyze the random linear coding scheme by examining the
expected service time in a manner similar to the approach used for
the retransmission scheme in Section \ref{section:StabilityRetrans}.
Let $\widetilde{T}_n$ denote the service time for source $n$, which
is the number of slots that elapse from the transmission of the
first coded packet until the source has received acknowledgements
from both destinations. Since each completed service in the random
linear coding scheme results in $K$ packets being removed from the
queue, the average service rate is given by $\widetilde{\mu}_n = K/
E[\widetilde{T}_n]$.

The service time will be a random variable dependent upon the random
access transmission probabilities $p_n$, the reception probabilities
$q_{n|n}^{(m)}$ and $q_{n|1,2}^{(m)}$, and the number of coded
packets needed to decode. Let $N_{n}^{(m)}$ denote the number of
coded packets received from source $n$ at destination $m$ before
destination $m$ can decode the original $K$ packets. The
$N_{n}^{(m)}$ will be identically distributed over $n$ and $m$ since
all coded packets are generated in the same way according to
(\ref{eqn:RLCCodedPacket}). Additionally, $N_{n}^{(m)}$ will be
independent over $n$ since the two sources generate their coded
packets independently. However, $N_{n}^{(m)}$ will be correlated
over $m$, since a given source $n$ will generate and transmit the
same coded packets to both destinations $m=1,2$.

For a fixed value of $K$, let $F_K(j)$ denote the common cumulative
distribution function (cdf) of $N_{n}^{(m)}$, or the probability
that the number of coded packets needed for decoding is less than or
equal to $j$. With each coded packet it receives, the destination
collects a binary column of length $K$ in a matrix, where the column
consists of the coefficients $a_i$, $i=1,2,\hdots,K$. Then $F_K(j)$
is the probability that decoding can be performed (by solving a
system of linear equations, or equivalently, by Gaussian
elimination) if the matrix has at most $j$ columns. Thus
\begin{equation}
F_K(j) = \mbox{Pr}\{ \mbox{a random $K {\times} j$ binary matrix has
rank } K\}.
\end{equation}
Note that for $j<K$ the matrix cannot possibly have rank $K$ and
$F_K(j)=0$. For $j \geq K$, we can write an expression for $F_K(j)$
following the procedure in \cite{MacKay98} by first counting up the
number of $K \times j$ non-singular matrices, which is given by
\begin{equation}
(2^j-1)(2^j-2)(2^j-2^2)\hdots(2^j-2^{K-1}).
\end{equation}
In the product above, each term accounts for a row of the $K \times
j$ binary matrix and reflects that the row is neither zero, nor
equal to any of the previous rows, nor equal to any linear
combination of previous rows. Since the total number of $K \times j$
binary matrices is $2^{jK}$, we obtain
\begin{equation}
F_K(j) = \begin{cases} \prod_{i=0}^{K-1} (1-2^{-j+i}) & j\geq K \\ 0
& j<K \end{cases} \label{eqn:cdfN}
\end{equation}
The probability mass function (pmf) of $N_{n}^{(m)}$, or the
probability that decoding can be performed when the destination has
received exactly $j$ columns and no fewer, is given by
\begin{equation}
f_K(j) = F_K(j) - F_K(j-1). 
\label{eqn:pmfNascdfN}
\end{equation}
From (\ref{eqn:cdfN}) and (\ref{eqn:pmfNascdfN}) the expected value
of $N_{n}^{(m)}$ can be computed, where $E[N_{n}^{(m)}] \geq K$. The
ratio
\begin{eqnarray}
\frac{E[N_{n}^{(m)}]}{K} & = & \prod_{i=0}^{K-1} (1-2^{-K+i}) +
\frac{1}{K} \sum_{j=K+1}^{\infty} j \left( \prod_{i=0}^{K-1}
(1-2^{-j+i}) - \prod_{i=0}^{K-1} (1-2^{-j+1+i}) \right) \\
 & \longrightarrow & 1 \quad \mbox{
 as } K \rightarrow \infty
\end{eqnarray}
for all $n,m=1,2$.

With the distribution of $N_{n}^{(m)}$ characterized, we can now
describe the service time from source $n$ to destination $m$, which
we denote $\widetilde{T}_n^{(m)}$. The number of slots needed for
the successful reception of each coded packet will be geometrically
distributed and in total $N_{n}^{(m)}$ coded packets must be
received. Then $\widetilde{T}_n^{(m)}$ can be modeled as the sum of
$N_{n}^{(m)}$ independent geometrically distributed random
variables.
\begin{equation}
\widetilde{T}_n^{(m)} = g_{n,1}^{(m)} + g_{n,2}^{(m)} + \hdots +
g_{n,N_{n}^{(m)}}^{(m)}
\end{equation}
In the above expression, $g_{n,i}^{(m)}$,
$i=1,2,\hdots,N_{n}^{(m)}$, will be geometrically distributed with a
parameter that depends on the reception probabilities and on the
assumption of whether the other source is backlogged. For instance,
when source 2 is backlogged
\begin{equation}
g_{1,i}^{(m)} \sim \mbox{geom} \left(p_1(1-p_2)q_{1|1}^{(m)} +
p_1p_2 q_{1|1,2}^{(m)} \right), \quad i=1,2,\hdots,N_{1}^{(m)}
\end{equation}
and when source 2 is empty,
\begin{equation}
g_{1,i}^{(m)} \sim \mbox{geom} \left(p_1q_{1|1}^{(m)} \right), \quad
i=1,2,\hdots,N_{1}^{(m)}.
\end{equation}
The total service time is given as
\begin{equation}
\widetilde{T}_n = \max_{m} \widetilde{T}_n^{(m)}.
\label{eqn:RLCmaxTnm}
\end{equation}

As in the case of random access with retransmissions, we argue that
the stable throughput region for random access with random linear
coding will be outer bounded by the Shannon capacity region. In the
case that source 2 is backlogged, the expected service time for
source 1 is now bounded as
\begin{eqnarray}
E[\widetilde{T}_1] = E[\max_m \widetilde{T}_1^{(m)}] &
\stackrel{(a)} \geq & \max_m
E[\widetilde{T}_1^{(m)}] \\
 & \stackrel{(b)} = & \max_m  \frac{E[N_{1}^{(m)}]}{p_1 (1-p_2) q_{1|1}^{(m)} + p_1 p_2
q_{1|1,2}^{(m)}} \\
 & \stackrel{(c)}= & \frac{E[N_{1}^{(m)}]}{\min_m p_1 (1-p_2) q_{1|1}^{(m)} + p_1 p_2
q_{1|1,2}^{(m)}}
\end{eqnarray}
where $(a)$ again follows from Jensen's inequality, $(b)$ holds
since $g_{1,i}^{(m)}$ are independent, identically distributed, and
$(c)$ holds since $N_{1}^{(m)}$ is identically distributed over $m$,
meaning that $E[N_{1}^{(1)}]=E[N_{1}^{(2)}]$. The backlogged service
rate $\widetilde{\mu}_{1b}$ is bounded as
\begin{eqnarray}
\widetilde{\mu}_{1b} = \frac{K}{E[\widetilde{T}_1]} & \leq &
\frac{K}{E[N_{1}^{(m)}]} \min_m p_1 (1-p_2)q_{1|1}^{(m)} +
p_1p_2q_{1|1,2}^{(m)} \\
& \leq & \min_m p_1 (1-p_2)q_{1|1}^{(m)} + p_1p_2q_{1|1,2}^{(m)}.
\end{eqnarray}
Then the Shannon capacity region outer bounds the stable throughput
region for random access with random linear coding.

Unfortunately a difficulty arises in finding the service rates
$\mu_{nb}$ and $\mu_{ne}$ in closed form from $E[\max_m
\widetilde{T}_n^{(m)}]$. This difficulty arises for a number of
reasons: $\widetilde{T}_n^{(1)}$ and $\widetilde{T}_n^{(2)}$ are not
independent, and $\widetilde{T}_n^{(m)}$ is distributed according to
a composite distribution function, for which the pdf is not easily
expressed in closed form. In fact, even if these two difficulties
were to be removed, $E[\max_m\widetilde{T}_n^{(m)}]$ cannot be
easily handled. For instance, let us assume that
$\widetilde{T}_n^{(1)}$ and $\widetilde{T}_n^{(2)}$ are independent
and that $N_n^{(m)}=n^{(m)}$ are deterministic (which means that the
pdf is no longer composite). In that case, $\widetilde{T}_n^{(m)}$
is the sum of $n^{(m)}$ iid geometric random variables, meaning that
$\widetilde{T}_n^{(m)}$ follows a negative binomial distribution.
Let us further make the assumption that
$q_{n|n}^{(1)}=q_{n|n}^{(2)}$ and $q_{n|1,2}^{(1)}=q_{n|1,2}^{(2)}$,
which means that $\widetilde{T}_n^{(m)}$ are identically distributed
over $m$. In this very simplified case, $E[\max_m
\widetilde{T}_n^{(m)}]$ is the expected maximum of two iid negative
binomial random variables. The computation of this expected value is
treated in \cite{GrabnerProdinger97}, and the result involves a
periodic function which is approximated by a Fourier series. Thus,
even in this very simplified case, we can at best approximate
$E[\max_m \widetilde{T}_n^{(m)}]$, and this approximation must be
computed numerically.

\subsection{Markov chain approach} \label{subsection:RLCMarkovChain}

As an alternative to the analysis presented above, we now develop a
Markov chain model which allows us to find the queueing service
rates. For a given source node, we set up a vector Markov chain with
state $(i,j,k)$ corresponding to the number of linearly independent
coded packets that have been received from the source node. In this
model, $i$ represents the number of linearly independent coded
packets that have been received at destination 1, and $j$ represents
the number of linearly independent packets that have been received
at destination 2. Since the coded packets are generated by the same
source, some of the coded packets received at destination 1 may also
have been received at destination 2, and $k$ represents the number
of such packets, where $k \leq \min(i,j)$. The variable $k$ allows
us to track the correlation between $N_n^{(1)}$ and $N_n^{(2)}$,
which was a difficulty in our previous approach described above. The
Markov chain evolves in discrete time over the time slots in our
system model.

The state space of the Markov chain is the three-dimensional
discrete set of points $[0, K]^3$. There are $K+1$ absorbing states
given by $(K,K,k)$, $0 \leq k \leq K$, which represent the reception
of $K$ linearly independent coded packets at both destinations, for
which the service of $K$ packets at the source has been completed.
Transitions in the Markov chain can only occur ``upward'',
corresponding to the reception of a new linearly independent packet,
and a transition results in an increase of the indices $i,j,k$ by at
most 1, meaning that at most 1 new linearly independent packet can
be received in a slot. We use the notation  $(i_1,j_1,k_1)
\rightarrow (i_2,j_2,k_2)$ to denote the transition from state
$(i_1,j_1,k_1)$ to state $(i_2,j_2,k_2)$.

The Markov chain is irreducible and aperiodic, and because it has a
finite state space, a stationary distribution exists. Let
$\pi_{i,j,k}$ denote the steady-state probability of $(i,j,k)$. The
steady-state probabilities are found by solving the set of equations
\begin{equation}
\pi_{i_1,j_1,k_1} = \sum_{(i_2,j_2,k_2)} \pi_{i_2,j_2,k_2}
\mbox{Pr}( (i_2,j_2,k_2) \rightarrow (i_1,j_1,k_1) )
\end{equation}
and
\begin{equation}
\sum_{(i,j,k)} \pi_{i,j,k} = 1.
\end{equation}
The service rate $\widetilde{\mu}_{n}$ is equal to $K$ times the
probability of transitioning into an absorbing state $(K,K,k)$, $0
\leq k \leq K$. There are only a few ways to transition into an
absorbing state; let ${\cal A}_k$ denote the set of states which
have a one-step transition into the absorbing state $(K,K,k)$. For
$k \in [0,K-1]$ we have
\begin{equation}
{\cal A}_k = \left\{(K{-}1,K,k), (K{-}1,K,k{-}1), (K,K{-}1,k),
(K,K{-}1,k{-}1), (K{-}1,K{-}1,k{-}1) \right\},
\end{equation}
and for $k=K$,
\begin{equation}
{\cal A}_K = \left\{(K-1,K,K-1), (K,K-1,K-1), (K-1,K-1,K-1)
\right\}.
\end{equation}
Note that we define ${\cal A}_K$ in this way since the states
$(K-1,K,K)$ and $(K,K-1,K)$ violate $k \leq \min(i,j)$. The service
rate for source $n$ is given by
\begin{equation}
\widetilde{\mu}_n = K \sum_{k=0}^{K} \sum_{(i,j,k) \in {\cal A}_k}
\pi_{i,j,k} \mbox{Pr}((i,j,k) \rightarrow (K,K,k)).
\label{eqn:muRLCMC}
\end{equation}
The transition probabilities $(i_1,j_1,k_1) \rightarrow
(i_2,j_2,k_2)$ for source $n$ can be written assuming that the other
source is either backlogged or empty, leading to the service rates
$\widetilde{\mu}_{nb}$ and $\widetilde{\mu}_{ne}$.

As an example, consider the transition $(i,j,k) \rightarrow
(i+1,j,k)$ in the Markov chain for source 1 when source two is
backlogged. Assume first that source two does not transmit, which
happens with probability $1-p_2$. Then there are two ways for the
transition $(i,j,k) \rightarrow (i+1,j,k)$ to occur. First,
destination 2 could receive no packet, which happens with
probability $1-q_{1|1}^{(2)}$, while destination 1 receives a coded
packet which is neither an all-zero packet nor equal to any linear
combination of the $i$ packets it has already received, which
happens with probability $q_{1|1}^{(1)}(1-2^{i}2^{-K})$.
Alternatively, both destinations could receive a coded packet, but
that packet is either the all zero packet or some linear combination
of the packets that have been received by destination 2 and {\it
not} by destination 1. This happens with probability $q_{1|1}^{(1)}
q_{1|1}^{(2)} (2^j - 2^k)2^{-K}$. The same two alternatives are
possible in the case that source 2 does transmit, which happens with
probability $p_2$, except that the reception probabilities are now
given by $q_{1|1,2}^{(m)}$. Then the transition $(i,j,k) \rightarrow
(i+1,j,k)$ for source 1 when source 2 is backlogged occurs with
probability
\begin{multline}
p_1 \left[ \overline{p}_2
\left\{q_{1|1}^{(1)}(1-q_{1|1}^{(2)})(1{-}2^{i}2^{-K}) +
q_{1|1}^{(1)}q_{1|1}^{(2)}(2^j-2^k)2^{-K} \right\} \right. \\
\left. \hspace{0.25in}  {+} p_2 \left\{
q_{1|1,2}^{(1)}(1-q_{1|1,2}^{(2)})(1{-}2^i 2^{-K}) +
q_{1|1,2}^{(1)}q_{1|1,2}^{(2)}(2^j-2^k)2^{-K}\right\} \right].
\end{multline}
The same transition probability can be used when source 2 is empty
by setting $p_2=0$. Similar arguments can be used to find all
transition probabilities for our Markov chain model; we have stated
those probabilities in Appendix \ref{app:RLCMCtransprob}.
Ultimately, we would like to find closed-form expressions for the
service rates $\widetilde{\mu}_{nb}$ and $\widetilde{\mu}_{ne}$, but
due to the size of the state-space, this is a difficult task.
Instead we have computed some numerical examples based on the Markov
chain model presented above, and those are presented next.

\section{Numerical examples} \label{section:NumericalExamples}

\begin{figure}
\centering
\includegraphics[width=0.7\textwidth]{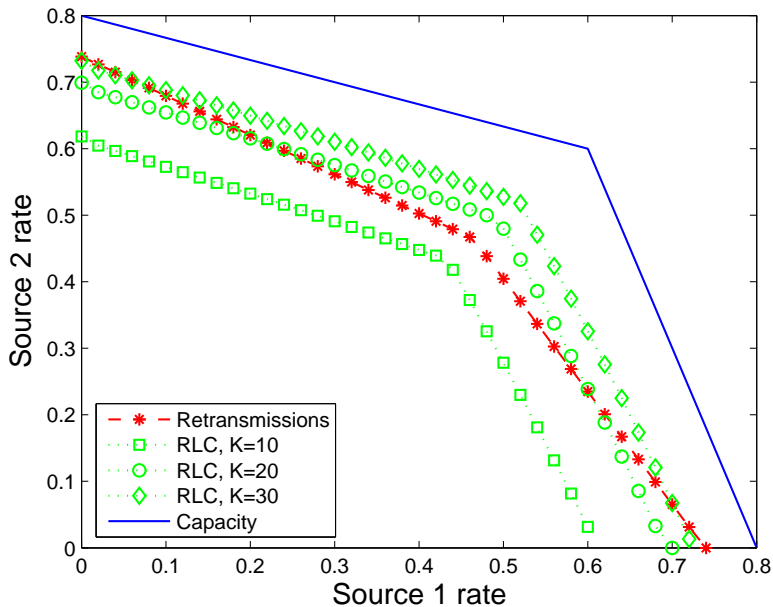}
\caption{Closure of the stable throughput and capacity regions for a
channel with reception probabilities $q_{1|1}^{(1)}=q_{2|2}^{(2)} =
0.8$, $q_{1|1}^{(2)}=q_{2|2}^{(1)}=0.7$, and
$q_{1|1,2}^{(m)}=q_{2|1,2}^{(m)}=0.6, m=1,2$. The stable throughput
region for random linear coding is abbreviated RLC.}
\label{fig:resultsStrongMPR}
\end{figure}

We have computed a number of numerical examples of the Shannon
capacity region and the stable throughput regions for
retransmissions and random linear coding. The results for random
linear coding have been computed with the service rates given by
Equation (\ref{eqn:muRLCMC}). Fig. \ref{fig:resultsStrongMPR} shows
results for a ``good'' channel with relatively large reception
probabilities while Fig. \ref{fig:resultsWeakMPR} shows results for
a ``poor'' channel with smaller reception probabilities. In both
figures, we have plotted the stable throughput for random linear
coding (abbreviated ``RLC'') with various values of $K$. The results
show that the Shannon capacity region is strictly larger than the
stable throughput regions for both the retransmissions and random
linear coding schemes. Additionally, the stable throughput region
for random linear coding grows with $K$ and appears to approach the
capacity as $K \rightarrow \infty$.

The random linear coding scheme does not necessarily outperform the
retransmission scheme. For small values of $K$, the coding scheme is
inefficient in the sense that the ratio $K/E[N_n^{(m)}]$ is small.
This inefficiency is largely due to the fact that an all-zero coded
packet can be generated and transmitted; this possibility is not
precluded in (\ref{eqn:RLCCodedPacket}) and occurs more often for
small values of $K$. The inefficiency of the coding scheme means
that the retransmission of $K$ packets requires a service time that
is less than $E[N_n^{(m)}]$ for small values of $K$. This effect
seems to become more pronounced as the channel improves, since for a
``good'' channel, packets are more often received correctly and do
not need to be retransmitted.

\begin{figure}
\centering
\includegraphics[width=0.7\textwidth]{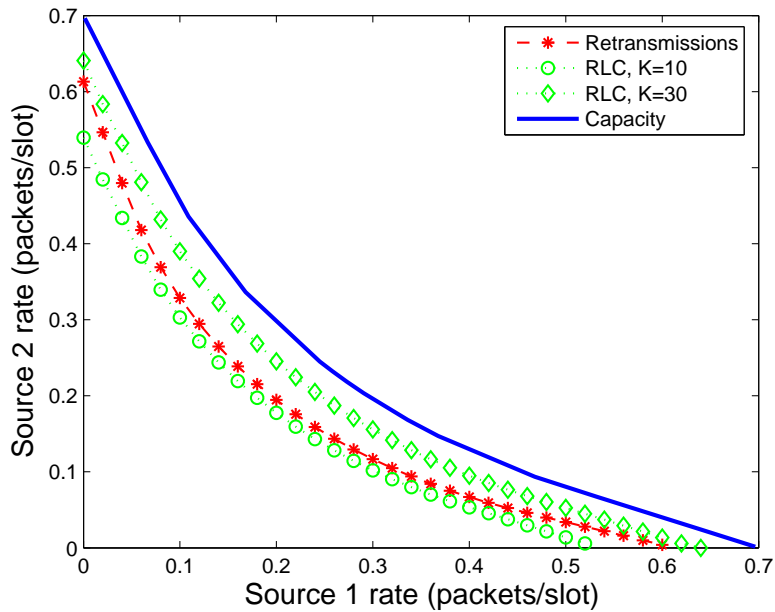}
\caption{Closure of the stable throughput and capacity regions for a
channel with reception probabilities $q_{1|1}^{(1)}=q_{2|2}^{(2)} =
0.8$, $q_{1|1}^{(2)}=q_{2|2}^{(1)}=0.7$, and
$q_{1|1,2}^{(m)}=q_{2|1,2}^{(m)}=0.2, m=1,2$. The stable throughput
region for random linear coding is abbreviated RLC.}
\label{fig:resultsWeakMPR}
\end{figure}

\section{Discussion} \label{section:conclusion}

One conclusion to draw from our work is that for sufficiently large
blocklengths, given by $K$ in our model, coding over packets in a
queue yields a higher stable multicast throughput than simply
retransmitting lost packets. In light of the recent developments on
network coding and fountain codes, this result is not surprising.
However, we have also shown that if the blocklength is too small and
the code is not designed appropriately, then a retransmission scheme
can provide a higher stable throughput than coding. Furthermore, we
have combined random access with random linear coding to yield an
efficient multicast scheme that can operate in a completely
distributed manner. We have shown that the merging of random access
with random linear coding results in good performance in the sense
that the stable throughput approaches the Shannon capacity. In the
process, we have presented a model that represents random coding of
packets in a queue.

Another significant outcome of our work is that we have provided an
example of a transmission policy for which the stable throughput
region {\it does not} coincide with the Shannon capacity region of a
random access system. This outcome contradicts a widely-held (yet
unproved) belief that the stability and capacity regions coincide
for random access. This result sheds further light on the relation
between the stable throughput and the Shannon capacity as
representations for the rate of reliable communication.

%

%

\appendices

\section{Proof of Theorem \ref{thrm:capreg}} \label{app:proofcapreg}
Let $P_m^t$ denote the error probability for the $t^{th}$ message
from the sources at receiver $m$.
\begin{equation*}
P_m^t \triangleq \mbox{Pr} \left\{ \bigcup_n \{J_n^t \neq
\hat{J}_{nm}^t \} \right\}.
\end{equation*}
The following lemma provides a condition equivalent to $P_e^t
\rightarrow 0$ and is used in the proof of Theorem
\ref{thrm:capreg}.
\begin{lemma}
The average error $P_e^t \rightarrow 0$ if and only if $\max_{m}
P_m^t \rightarrow 0$. \label{lemma:errorcond}
\end{lemma}
\begin{proof}
The average error $P_e^t$ can be upper bounded by the union bound as
follows: $P_e^t  \leq  P_1^t + P_2^t \leq 2 \max_{m} P_m^t$. A
similar lower bound also holds, namely $P_e^t  \geq  \max_{m}
P_m^t$. Thus $\max_{m} P_m^t \leq P_e^t \leq 2 \max_{m} P_m^t$ and
the result follows.
\end{proof}

Achievability for our system is shown by first establishing
achievability for the MAC $W_m$. This is shown in
\cite{HuiHumblet85} and \cite{Poltyrev83}; the approach presented in
\cite{HuiHumblet85} is summarized here. Each codeword symbol in the
codebook for source $n$ is generated according to the distribution
$P(x_n)$, independently over codeword symbols and independently over
messages. The following two properties are assumed.
\begin{enumerate}
\item[(I)] The codewords $x_n(1)$ and $x_n(2)$ are reserved for use
as preambles. A preamble is sent after every sequence of $T$
messages and $x_n(1)$ and $x_n(2)$ are used as preambles in an
alternating fashion. In \cite{HuiHumblet85} it is shown that by
using the preamble, the receiver can synchronize on codeword
boundaries with arbitrarily small probability of synchronization
error.
\item[(II)] In a sequence of $T+1$ messages (including a preamble),
no messages are repeated. As a result, any two disjoint subsets of
$N(T+1)$ codeword symbols (corresponding to $T+1$ messages) are
independent. For $T \ll 2^{NR_n}$ the resulting loss in rate is
negligible.
\end{enumerate}
By observing the channel outputs, the decoder $\phi_{1m}$ can detect
the preambles $x_1(1)$ and $x_1(2)$ to determine that the output
symbols in between correspond to inputs
$x_1(j_1^1),x_1(j_1^2),\hdots,x_1(j_1^T)$. Let $x_1^+$ denote the
sequence of $N(T+1)$ symbols corresponding to
$x_1(j_1^1),x_1(j_1^2),\hdots,x_1(j_1^T)$ augmented by portions of
the preambles $x_1(1)$ and $x_2(2)$. At the channel output, $x_1^+$
will overlap with a sequence $x_2^+$ consisting of $N(T+1)$ symbols
input by source 2, including $N$ preamble symbols. Let
$y_m^{N(T+1)}$ denote the output sequences corresponding to $x_1^+$
and $x_2^+$ at destination $m$. The decoder $\phi_{1m}$ outputs the
unique sequence of messages
$\hat{j}_{1m}^1,\hat{j}_{1m}^2,\hdots,\hat{j}_{1m}^T$ that lies in
the set of typical $(x_1^+,x_2^+,y_m^{N(T+1)})$ sequences. With this
approach it is shown in \cite{HuiHumblet85} that $\mbox{Pr}\{J_1^t
\neq J_{1m}^{t}\} \rightarrow 0$ for all $t$. A similar technique
can be used by decoder $\phi_{2m}$ to show that $\mbox{Pr}\{J_2^t
\neq J_{2m}^{t}\} \rightarrow 0$. Then by the union bound, $P_m^t
\to 0$ for all $t$. Finally, if the rate pair $(R_1,R_2)$ lies in
the intersection of the achievable rates for MACs $W_1$ and $W_2$,
then $\max_m P_m^t \rightarrow 0$ and thus $P_e^t \rightarrow 0$ for
all $t$ by Lemma \ref{lemma:errorcond}.

The converse for the MAC $W_m$ is shown in \cite{HuiHumblet85} and
\cite{Poltyrev83} by using Fano's inequality, the data processing
inequality, and the concavity of mutual information. Then $\max_m
P_m^t \rightarrow 0$ implies that the rate pair $(R_1,R_2)$ must lie
within the intersection of the capacity regions of $W_1$ and $W_2$.

\section{Markov chain analysis of random linear coding} \label{app:RLCMCtransprob}
In the Markov chain analysis of Section
\ref{subsection:RLCMarkovChain}, the state $(i,j,k)$ represents $i$
linearly independent coded packets received at destination 1, $j$
linearly independent coded packets received at destination 2, and
$k$ coded packets which have been received at both destinations, $k
\leq \min(i,j)$. When source 2 is backlogged, the non-zero
transition probabilities for source 1 are given as follows for
$i,j=0,1,\hdots,K$.
\begin{IEEEeqnarray*}{rl}
(i,j,k) \rightarrow (i,j,k) : & \overline{p}_1 {+} p_1 \left[ \overline{p}_2 \left\{ (1{-}q_{1|1}^{(1)})(1{-}q_{1|1}^{(2)}) {+} (1{-}q_{1|1}^{(1)})q_{1|1}^{(2)}2^{j}2^{-K} \right. \right. \\ & \left. \left. \hspace{0.75in} {+} q_{1|1}^{(1)}(1-q_{1|1}^{(2)})2^i 2^{-K} {+} q_{1|1}^{(1)} q_{1|1}^{(2)}2^k 2^{-K} \right\} \right. \\ & \left. \hspace{0.5in} {+} p_2 \left\{ (1{-}q_{1|1,2}^{(1)})(1{-}q_{1|1,2}^{(2)}) {+} (1{-}q_{1|1,2}^{(1)})q_{1|1,2}^{(2)}2^j2^{-K} \right. \right. \\ & \hspace{0.75in} \left. \left. {+} q_{1|1,2}^{(1)}(1-q_{1|1,2}^{(2)})2^i2^{-K} {+} q_{1|1,2}^{(1)} q_{1|1,2}^{(2)}2^k 2^{-K} \right\} \right]\\
(i,j,k) \rightarrow (i+1,j,k): & p_1 \left[ \overline{p}_2 \left\{q_{1|1}^{(1)}(1-q_{1|1}^{(2)})(1{-}2^i2^{-K}) + q_{1|1}^{(1)}q_{1|1}^{(2)}(2^j-2^k)2^{-K} \right\} \right. \\ & \left. \hspace{0.25in}  {+} p_2 \left\{ q_{1|1,2}^{(1)}(1-q_{1|1,2}^{(2)})(1{-}2^i2^{-K}) + q_{1|1,2}^{(1)}q_{1|1,2}^{(2)}(2^j-2^k)2^{-K}\right\} \right]\\
(i,j,k) \rightarrow (i,j+1,k): & p_1 \left[ \overline{p}_2 \left\{ (1-q_{1|1}^{(1)})q_{1|1}^{(2)}(1{-}2^j2^{-K}) + q_{1|1}^{(1)}q_{1|1}^{(2)}(2^i-2^k)2^{-K} \right\} \right. \\ & \left. \hspace{0.25in} {+} p_2 \left\{ (1-q_{1|1,2}^{(1)})q_{1|1,2}^{(2)}(1{-}2^j2^{-K}) + q_{1|1,2}^{(1)}q_{1|1,2}^{(2)}(2^i-2^k)2^{-K}  \right\} \right]\\
(i,j,k) \rightarrow (i+1,j+1,k+1): & p_1 \left[ \overline{p}_2 \left\{ q_{1|1}^{(1)} q_{1|1}^{(2)} (1{-}(2^i+2^j-2^k)2^{-K}) \right\} \right. \\ & \left. \hspace{0.25in}{+} p_2 \left\{ q_{1|1,2}^{(1)} q_{1|1,2}^{(2)} (1{-}(2^i+2^j-2^k)2^{-K}) \right\} \right]\\
\end{IEEEeqnarray*}
\begin{IEEEeqnarray*}{rl}
(i,K,k) \rightarrow (i,K,k): & \overline{p}_1 {+} p_1 \left[ \overline{p}_2 \left\{ (1{-}q_{1|1}^{(1)}) {+} q_{1|1}^{(1)}2^i 2^{-K}\right\}  {+} p_2\left\{ (1{-}q_{1|1,2}^{(1)}) {+} q_{1|1,2}^{(1)}2^i 2^{-K}\right\} \right] \\
(i,K,k) \rightarrow (i+1,K,k): & p_1 \left[ \overline{p}_2 \left\{ q_{1|1}^{(1)}(1{-}(2^i+K-k)2^{-K}) \right\} {+} p_2 \left\{ q_{1|1,2}^{(1)}(1{-}(2^i+K-k)2^{-K}) \right\} \right]\\
(i,K,k) \rightarrow (i+1,K,k+1): & p_1 \left[ \overline{p}_2 \left\{q_{1|1}^{(1)}(K-k)2^{-K} \right\} {+} p_2 \left\{ q_{1|1,2}^{(1)}(K-k)2^{-K} \right\} \right]\\
\end{IEEEeqnarray*}
\begin{IEEEeqnarray*}{rl}
(K,j,k) \rightarrow (K,j,k): & \overline{p}_1 {+} p_1 \left[ \overline{p}_2 \left\{ (1{-}q_{1|1}^{(2)}) {+} q_{1|1}^{(2)}2^j 2^{-K}\right\}  {+} p_2\left\{ (1{-}q_{1|1,2}^{(2)}) {+} q_{1|1,2}^{(2)}2^j 2^{-K}\right\} \right] \\
(K,j,k) \rightarrow (K,j{+}1,k): & p_1 \left[ \overline{p}_2 \left\{ q_{1|1}^{(2)}(1{-}(2^j+K-k)2^{-K}) \right\} {+} p_2 \left\{ q_{1|1,2}^{(2)}(1{-}(2^j+K-k)2^{-K}) \right\} \right]\\
(K,j,k) \rightarrow (K,j{+}1,k{+}1): & p_1 \left[ \overline{p}_2 \left\{q_{1|1}^{(2)}(K-k)2^{-K} \right\} {+} p_2 \left\{ q_{1|1,2}^{(2)}(K-k)2^{-K} \right\} \right]\\
\end{IEEEeqnarray*}

\section*{Acknowledgment}
The authors wish to thank Gerhard Kramer for helpful comments which
greatly improved the manuscript.


\bibliographystyle{IEEEtran}

\begin{thebibliography}{10}
\providecommand{\url}[1]{#1} \csname url@rmstyle\endcsname
\providecommand{\newblock}{\relax} \providecommand{\bibinfo}[2]{#2}
\providecommand\BIBentrySTDinterwordspacing{\spaceskip=0pt\relax}
\providecommand\BIBentryALTinterwordstretchfactor{4}
\providecommand\BIBentryALTinterwordspacing{\spaceskip=\fontdimen2\font
plus \BIBentryALTinterwordstretchfactor\fontdimen3\font minus
  \fontdimen4\font\relax}
\providecommand\BIBforeignlanguage[2]{{%
\expandafter\ifx\csname l@#1\endcsname\relax
\typeout{** WARNING: IEEEtran.bst: No hyphenation pattern has been}%
\typeout{** loaded for the language `#1'. Using the pattern for}%
\typeout{** the default language instead.}%
\else \language=\csname l@#1\endcsname \fi #2}}

\bibitem{CoverThomas}
T.~M. Cover and J.~A. Thomas, \emph{Elements of Information
Theory}.\hskip 1em
  plus 0.5em minus 0.4em\relax New York: Wiley, 1991.

\bibitem{Abramson70}
N.~Abramson, ``The aloha system - another alternative for computer
  communications,'' in \emph{Fall Joint Computer Conference (AFIPS)}, vol.~37,
  1970, pp. 281--285.

\bibitem{TsybakovMikhailov}
B.~Tsybakov and V.~Mikhailov, ``Ergodicity of a slotted aloha
system,''
  \emph{Probl. Inform. Tran}, vol.~15, pp. 301--312, 1979.

\bibitem{RaoAE88}
R.~Rao and A.~Ephremides, ``On the stability of interacting queues
in a
  multiple-access system,'' \emph{{IEEE} Trans. Inform. Theory}, vol.~34, pp.
  918--930, September 1988.

\bibitem{Szpa94}
W.~Szpankowski, ``Stability conditions for some distributed systems:
Buffered
  random access systems,'' \emph{Advances in Applied Probability}, vol.~26, pp.
  498--515, June 1994.

\bibitem{LuoAE99}
W.~Luo and A.~Ephremides, ``Stability of n interacting queues in
random-access
  systems,'' \emph{{IEEE} Trans. Inform. Theory}, vol.~45, July 1999.

\bibitem{NawareEtAl03}
V.~Naware, G.~Mergen, and L.~Tong, ``Stability and delay of finite
user slotted
  aloha with multipacket reception,'' \emph{{IEEE} Trans. Inform. Theory},
  vol.~51, pp. 2636--2656, July 2005.

\bibitem{MasseyMathys}
J.~L. Massey and P.~Mathys, ``The collision channel without
feedback,''
  \emph{{IEEE} Trans. Inform. Theory}, vol. IT-31, pp. 192--204, March 1985.

\bibitem{Hui}
J.~Y.~N. Hui, ``Multiple accessing for the collision channel without
  feedback,'' \emph{{IEEE} J. Select. Areas Commun.}, vol. SAC-2, pp.
  575--582j, July 1984.

\bibitem{Poltyrev83}
G.~S. Poltyrev, ``Coding in an asynchronous multiple access
channel,''
  \emph{Problems on Information Transmission}, vol.~19, pp. 184--191 (12--21 in
  Russian version), 1983.

\bibitem{HuiHumblet85}
J.~Y.~N. Hui and P.~A. Humblet, ``The capacity of the totally
asynchronous
  multiple-access channel,'' \emph{{IEEE} Trans. Inform. Theory}, vol.~31, pp.
  207--216, March 1985.

\bibitem{Tinguely}
S.~Tinguely, M.~Rezaeian, and A.~Grant, ``The collision channel with
  recovery,'' \emph{{IEEE} Trans. Inform. Theory}, vol. IT-51, pp. 3631--3638,
  October 2005.

\bibitem{Anantharam}
V.~Anantharam, ``Stability region of the finite-user slotted aloha
protocol,''
  \emph{{IEEE} Trans. Inform. Theory}, vol.~37, pp. 535--540, May 1991.

\bibitem{RockeyAE_ITTrans06}
J.~Luo and A.~Ephremides, ``On the throughput, capacity and
stability regions
  of random multiple access,'' \emph{{IEEE} Trans. Inform. Theory}, vol.~52,
  pp. 2593--2607, June 2006.

\bibitem{Ahlswede74}
R.~Ahlswede, ``The capacity region of a channel with two senders and
two
  receivers,'' \emph{Annals of Probability}, vol.~2, pp. 805--814, 1974.

\bibitem{NetCod1}
R.~Ahlswede, N.~Cai, S.-Y.~R. Li, and R.~W. Yeung, ``Network
information
  flow,'' \emph{{IEEE} Trans. Inform. Theory}, vol.~46, pp. 1204--1216, July
  2000.

\bibitem{NetCod2}
S.-Y.~R. Li, R.~W. Yeung, and N.~Cai, ``Linear network coding,''
\emph{{IEEE}
  Trans. Inform. Theory}, vol.~49, pp. 371--381, February 2003.

\bibitem{Luby02}
M.~Luby, ``Lt codes,'' in \emph{IEEE Symposium on the Foundations of
Computer
  Science}, 2002, pp. 271--280.

\bibitem{Shokrollahi06}
A.~Shokrollahi, ``Raptor codes,'' \emph{{IEEE} Trans. Inform.
Theory}, 2006, to
  appear. Available at algo.epfl.ch.

\bibitem{Gallager76}
R.~G. Gallager, ``Basic limits on protocol information in data
communication
  networks,'' \emph{{IEEE} Trans. Inform. Theory}, vol.~22, pp. 385--398, July
  1976.

\bibitem{Loynes62}
R.~Loynes, ``The stability of a queue with non-independent
interarrival and
  service times,'' \emph{Proc. Cambridge Phil. Soc.}, vol.~58, 1962.

\bibitem{ShraderEphremidesITTrans07StabThrpt}
B.~Shrader and A.~Ephremides, ``Random access broadcast: stability
and
  throughput analysis,'' \emph{{IEEE} Trans. Inform. Theory}, vol.~53, pp.
  2915--2921, August 2007.

\bibitem{SagduyuAE06}
Y.~E. Sagduyu and A.~Ephremides, ``On broadcast stability region in
random
  access through network coding,'' in \emph{Allerton Conference on
  Communication, Control, and Computing}, September 2006.

\bibitem{MacKay98}
D.~J.~C. MacKay, ``Fountain codes,'' in \emph{IEE Workshop on
Discrete Event
  Systems}, 1998.

\bibitem{GrabnerProdinger97}
P.~J. Grabner and H.~Prodinger, ``Maximum statistics of $n$ random
variables
  distributed by the negative binomial distribution,'' \emph{Combinatorics,
  Probability and Computing}, vol.~6, pp. 179--183, 1997.

\end{thebibliography}

%

%


\end{document}